\documentclass[sigconf,screen]{acmart}
\usepackage{listings}
\usepackage{xcolor}
\usepackage{tcolorbox}
\usepackage{xspace}
\usepackage{amsthm}
\usepackage{threeparttable}

\usepackage{balance}

\newtcbox{\mybox}[1][red]
  {on line, arc = 0pt, outer arc = 0pt,
    colback = #1!10!white, colframe = #1!50!black,
    boxsep = 0, left = 1pt, right = 1pt,
    boxrule = 0pt}
\AtBeginDocument{%
  \providecommand\BibTeX{{%
    \normalfont B\kern-0.5em{\scshape i\kern-0.25em b}\kern-0.8em\TeX}}}
\newcommand{\modify}[2]{\textcolor{black}{#2}}
\usepackage{bm}

\setcopyright{acmcopyright}
\acmPrice{15.00}
\acmDOI{10.1145/3468264.3468544}
\acmYear{2021}
\copyrightyear{2021}
\acmSubmissionID{fse21main-p92-p}
\acmISBN{978-1-4503-8562-6/21/08}
\acmConference[ESEC/FSE '21]{Proceedings of the 29th ACM Joint European Software Engineering Conference and Symposium on the Foundations of Software Engineering}{August 23--28, 2021}{Athens, Greece}
\acmBooktitle{Proceedings of the 29th ACM Joint European Software Engineering Conference and Symposium on the Foundations of Software Engineering (ESEC/FSE '21), August 23--28, 2021, Athens, Greece}
\begin{document}

\title{A Syntax-Guided Edit Decoder for Neural Program Repair}



\begin{abstract}
  Automated Program Repair (APR) helps improve the efficiency of software development and maintenance. Recent APR techniques use deep learning, particularly the encoder-decoder architecture, to generate patches.
  Though existing DL-based APR approaches have proposed different encoder architectures, the decoder remains to be the standard one, which generates a sequence of tokens one by one to replace the faulty statement.
  This decoder has multiple limitations: 1) allowing to generate syntactically incorrect programs, 2) inefficiently representing small edits, and 3) not being able to generate project-specific identifiers.
  
  In this paper, we propose \techname, a syntax-guided edit decoder with placeholder generation. \techname is novel in multiple aspects: 1) \techname generates edits rather than modified code, allowing efficient representation of small edits; 2) \techname is syntax-guided, with the novel provider/decider architecture to ensure the syntactic correctness of the patched program and accurate generation;  3) \techname generates placeholders that could be instantiated as project-specific identifiers later.

  We conduct experiments to evaluate \techname on 395 bugs from \textit{Defects4J v1.2}, 420 additional bugs from \textit{Defects4J v2.0}, 297 bugs from \textit{IntroClassJava} and 40 bugs from \textit{QuixBugs}. Our results show that \techname repairs 51 bugs on Defects4J v1.2, which achieves 21.4\% (9 bugs) improvement over the previous state-of-the-art approach for single-hunk bugs (TBar). Importantly, to our knowledge, \textit{\techname is the first DL-based APR approach that has outperformed the traditional APR approaches on this benchmark}. Furthermore, \techname repairs 19 bugs on the additional bugs from Defects4J v2.0, which is 137.5\% (11 bugs) more than TBar  and 850\% (17 bugs) more than SimFix. \techname also achieves 775\% (31 bugs) and 30.8\% (4 bugs) improvement on IntroClassJava and QuixBugs over the baselines respectively. These results suggest that \techname has better generalizability than existing APR approaches.
\end{abstract}


\author{Qihao Zhu}
\affiliation{%
  \institution{Key Laboratory of HCST, MoE \\DCST, Peking University}
   \country{Beijing, China}
  }
\email{Zhuqh@pku.edu.cn}

\author{Zeyu Sun}
\affiliation{%
  \institution{Key Laboratory of HCST, MoE \\DCST, Peking University}
  \country{Beijing, China}
  }
\email{szy_@pku.edu.cn}

\author{Yuan-an Xiao}
\affiliation{%
  \institution{Key Laboratory of HCST, MoE \\DCST, Peking University}
   \country{Beijing, China}
  }
\email{xiaoyuanan@pku.edu.cn}

\author{Wenjie Zhang}
\affiliation{%
  \institution{Key Laboratory of HCST, MoE \\DCST, Peking University}
   \country{Beijing, China}
  }
\email{zhang_wen_jie@pku.edu.cn}

\author{Kang Yuan}
\affiliation{%
  \institution{Stony Brook University}
  \country{New York, US}
  }
\email{kang.yuan@stonybrook.edu}

\author{Yingfei Xiong}
\authornote{Corresponding author. \\HCST: High Confidence Software Technologies.}
\affiliation{%
  \institution{Key Laboratory of HCST, MoE \\DCST, Peking University}
   \country{Beijing, China}
  }
\email{xiongyf@pku.edu.cn}

\author{Lu Zhang}
\affiliation{%
  \institution{Key Laboratory of HCST, MoE \\DCST, Peking University}
  \country{Beijing, China}
  }
\email{zhanglucs@pku.edu.cn}
\newcommand{\techname}{Recoder\xspace}

\maketitle

\section{Introduction}
Automated program repair (APR) aims to reduce bug-fixing effort by generating patches to aid the developers. Due to the well-known problem of weak test suites~\cite{PatchPlausibility}, even if a patch passes all the tests, the patch still has a high probability of being incorrect. To overcome this problem, existing approaches have used different means to guide the patch generation. A typical way is to learn from existing software repositories, such as learning patterns from existing patches~\cite{6606626,2018Shaping,jiang2019inferring,liu2019tbar,long2017automatic,bader2019getafix,DBLP:conf/icse/RolimSDPGGSH17}, and using program code to guide the patch generation~\cite{xiong2017precise,long2016automatic,2018Shaping,ELIXIR,xiong2018learning}. 

Deep learning is known as a powerful machine learning approach. Recently, a series of research efforts have attempted to use deep learning (DL) techniques to learn from existing patches for program repair~\citep{8827954,DBLP:journals/corr/abs-1812-07170,9000077,dlfix}. A typical DL-based approach generates a new statement to replace the faulty statement located by a fault localization approach. Existing DL-based approaches are based on the encoder-decoder architecture~\cite{Bahdanau2015NeuralMT}: the encoder encodes the faulty statement as well as any necessary code context into a fixed-length internal representation, and the decoder generates a new statement from it. For example, \citet{DBLP:journals/corr/abs-1812-07170} and \citet{9000077} adopt an existing neural machine translation architecture, NMT, to generate the bug fix; SequenceR~\cite{8827954} uses a sequence-to-sequence neural model with a copy mechanism; DLFix~\citep{dlfix} further treats the faulty statement as an AST rather than a sequence of tokens, and encodes the context of the statement. 

However, despite multiple existing efforts, DL-based APR approaches have not yet outperformed traditional APR approaches. Since deep learning has outperformed traditional approaches in many domains, in this paper we aim to further improve the performance of DL-based APR to understand whether we could outperform traditional APR using a DL-based approach. We observe that, though existing DL-based APR approaches have proposed different encoder architectures for APR, the decoder architecture remains to be the standard one, generating a sequence of tokens one by one to replace the original faulty program fragment. The use of this standard decoder significantly limits the performance of DL-based APR. Here we highlight three main limitations.

{\bf Limitation 1: Including syntactically incorrect programs in the patch space}. The goal of the decoder is to locate a patch from a patch space. The smaller the patch space is, the easier the task is. However, viewing a patch as a sequence of tokens unnecessarily enlarges the patch space, making the decoding task difficult. In particular, this space representation does not consider the syntax of the target programming language and includes many syntactically incorrect statements, which can never form a correct patch. 

{\bf Limitation 2: Inefficient representation of small edits}.
Many patches only modify a small portion of a statement, and re-generating the whole statement leads to an unnecessarily large patch space. For example, let us consider the patch of defect Closure-14 in the Defects4J benchmark~\cite{defects4j}, as shown in Figure~\ref{fig:pattern-de1}. This patch only changes one token in the statement, but under existing representation, it is encoded as a sequence of length 13. The program space containing this patch would roughly contain $n^{13}$ elements, where $n$ is the total number of tokens. On the other hand, let us consider a patch space including only one-token change edits. To generate that patch, only selecting a token in the faulty statement and a new token for replacement is needed. This patch space contains only $mn$ elements, where $m$ is the number of tokens in the faulty statement. Therefore, the size of the patch space is significantly reduced.


{\bf Limitation 3: Not being able to generate project-specific identifiers}. Source code of programs often contains project-specific identifiers like variable names. Since it is impractical to include all possible identifiers in the patch space, existing DL-based APR approaches only generate identifiers that have frequently appeared in the training set. However, different projects have different sets of project-specific identifiers, and therefore only considering identifiers in the training set may exclude possible patches from the patch space. For example, Figure~\ref{fig:pattern-de2} shows the patch for defect Lang-57 in Defects4J. To generate this patch, we need to generate the identifier ``{\tt availableLocaleSet}'', which is a method name of the faulty class, and is unlikely to be included in the training set. As a result, existing DL-based approaches cannot generate patches like this. 

\smallskip
In this paper, we propose a novel DL-based APR approach, \techname, standing for \underline{\textbf{re}}pair de\underline{\textbf{coder}}. Similar to existing approaches, \techname is based on the encoder-decoder architecture. To address the limitations above, the decoder of \techname has following two novel techniques.

{\bf Novelty 1: Syntax-Guided Edit Decoding with Provider/Decider Architecture} (concerning limitation 1 \& 2).
To address limitation 2, the decoder component of \techname produces a sequence of edits rather than a new statement. Our edit decoder is based on the idea of the syntax-guided decoder in existing neural program generation approaches~\cite{treegen,DBLP:conf/acl/YinN17,DBLP:conf/acl/RabinovichSK17,DBLP:conf/aaai/SunZMXLZ19}. For an unexpanded non-terminal node in a partial AST, the decoder estimates the probability of each grammar rule to be used to expand the node. Based on this, the decoder selects the most probable sequence of rules to expand the start symbol into a full program using a search algorithm such as beam search. 
We observe that edits could also be described by a grammar. For example, the previous patch for defect Closure-14 could be described by the following grammar: 
\[\begin{array}{rcl}
    {\it Edit} &\rightarrow & {\it Insert} \mid {\it Modify} \mid \ldots \\
    {\it Modify} &\rightarrow & {\tt modify}({\it NodeID}, {\it NTS}) \\
\end{array}\]  
Here {\tt modify} represents replacing an AST subtree denoted by its root node ID ({\it NodeID}) in the faulty statement with a newly generated subtree ({\it NTS\footnote{``NTS'' stands for ``non-terminal symbol in an AST''.}}).
\begin{figure}
    \centering
    \includegraphics[width=\linewidth]{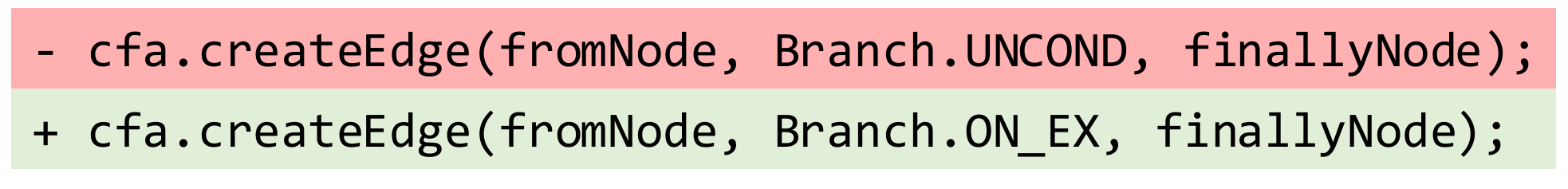}
    \vspace{-8mm}
    \caption{The Patch for Closure-14 in Defects4J }
    \vspace{-2mm}
    \label{fig:pattern-de1}
\end{figure}
\begin{figure}
    \centering
    \includegraphics[width=\linewidth]{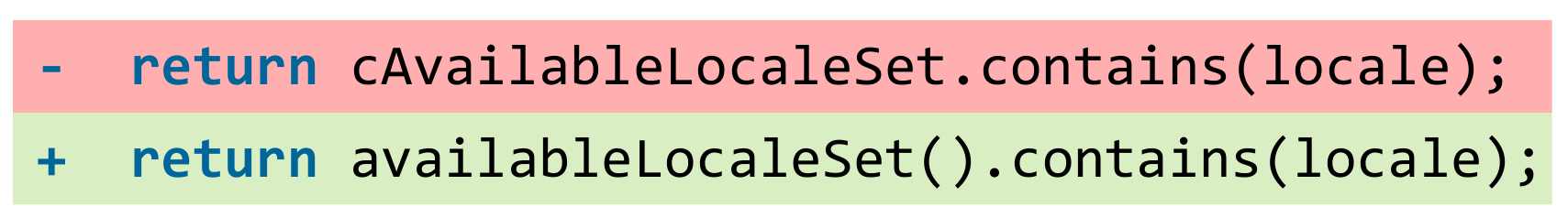}
    \vspace{-8mm}
    \caption{The Patch for Lang-57 in Defects4J}
    \vspace{-2mm}
    \label{fig:pattern-de2}
\end{figure}

However, directly applying the existing syntax-guided decoder to the grammar above would not form an effective program repair approach, because the choice of expanding different non-terminal nodes may need to be deduced along with different types of dependencies.
First, the expansion of some non-terminals depends on the local context, e.g., the choice of {\it NodeID} depends on the faulty statement, and the neural network needs to be aware of the local context to make a suitable choice. 
Second, to guarantee syntax correctness (limitation 1), dependency exists among the choices for expanding different non-terminal nodes, e.g., when {\it NodeID} expands to an ID pointing to a node with non-terminal {\it JavaExpr}, {\it NTS} should also expand to {\it JavaExpr} to ensure syntactic correctness. These choices cannot be effectively pre-defined, and thus the existing syntax-guided decoders, which only select among a set of pre-defined grammar rules, do not work here. 

\begin{figure}
    \centering
    \includegraphics[width=\linewidth]{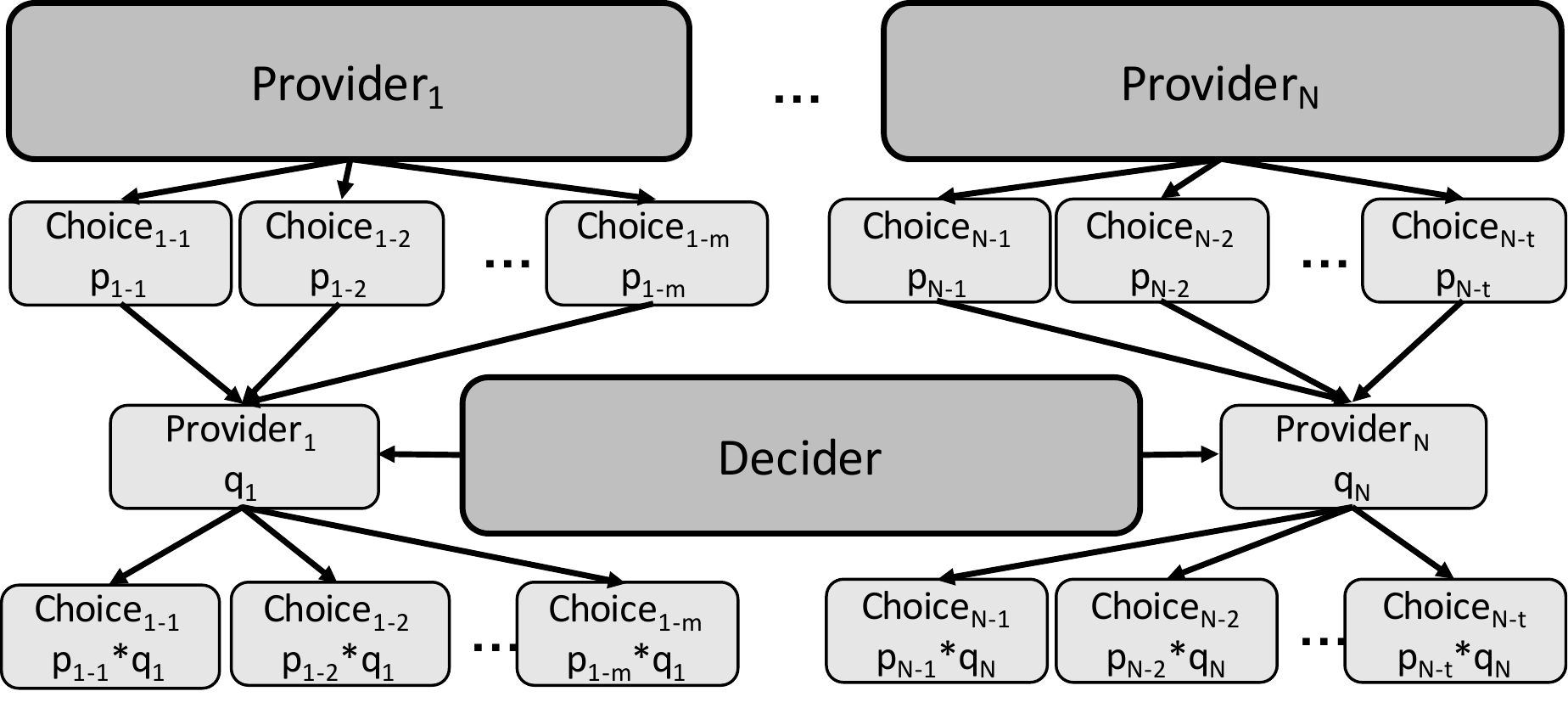}
    \vspace{-7mm}
    \caption{Provider/Decider Architecture}
    \vspace{-4mm}
    \label{fig:pd}
\end{figure}

To overcome these problems, \techname introduces a provider/decider architecture, as shown in Figure~\ref{fig:pd}. A provider is a neural component that provides a set of choices for expanding a non-terminal and estimates the probability $p_i$ of each choice. A basic provider is the rule predictor, which, similar to existing syntax-guided decoders, estimates the probability of each grammar rule to expand the node. Fixing the Closure-14 example needs another provider, namely the subtree locator, which estimates the probability of each subtree in the faulty statement to be replaced. On the other hand, the decider is a component that estimates the probability of $q_j$ using each provider. In this example, when expanding {\it Edit}, the probability of using the rule predictor is 1, and the probability of using the subtree locator is 0; when expanding {\it Modify}, the probability of using the rule predictor is 0 and the probability of using the subtree locator is 1 (the located subtree decides both the content of {\it NodeID} and the root symbol of {\it NTS}). Finally, the choices provided by all providers form the final list of choices, while the probability of each choice is the product of the probability predicted by its provider and the probability of the provider itself, i.e., $p_i*q_j$. 

In this example, for each non-terminal, we use the choices from only one provider, and thus the probabilities of providers are either 0 or 1. Later we will see that expanding some non-terminals requires comparing the choices of multiple providers, and the probabilities of providers could be a real number between 0 and 1.

{\bf Novelty 2: Placeholder Generation} (concerning limitation 3). 
To generate project-specific identifiers, a direct idea is to add another provider that selects an identifier from the local context. However, to implement such a provider, the neural component needs to access all of the name declarations within the current project. This is a difficult job, as the neural component could hardly encode all source code from the whole project. 

Instead of relying on the neural network to generate project-specific identifiers, in \techname the neural network generates placeholders for such identifiers, and these placeholders are instantiated with all feasible identifiers when applying the edits. A feasible identifier is an identifier compatible with constraints in the programming language, such as the type system. As for defect Lang-57 shown in Figure~\ref{fig:pattern-de2}, \techname first generates a placeholder for ``{\tt availableLocaleSet}'', and it will be replaced with all methods accessible in the local context that takes no arguments and returns an object with a member method ``{\tt contains}''. Each replacement forms a new patch. The key insight is that, when considering constraints in the programming language, the number of choices for replacing a placeholder with an identifier is small, and thus instantiating the placeholders with all possible choices is feasible.

To train the neural network to generate placeholders, we replace infrequent user-defined identifiers in the training set with placeholders. In this way, the neural network learns to generate placeholders for these identifiers.

\modify{}{Our experiment is conducted on four benchmarks: (1) 395 bugs from Defects4J v1.2 for comparison with existing approaches. (2) 420 additional bugs from  Defects4J v2.0, (3) 297 bugs from IntroClassJava, and (4) 40 bugs from QuixBugs to evaluate the generalizability of \techname. The results show that \techname correctly repairs 53 bugs on the first benchmark, which are 26.2\% (11 bugs) more than TBar~\cite{liu2019tbar} and 55.9\% (19 bugs) more than SimFix~\cite{2018Shaping}, two best-performing single-hunk APR approaches on Defects4J v1.2; \techname also correctly repairs 19 bugs on the second benchmark, which are 137.5\% (11 bugs) more than TBar and 850.0\% (17 bugs) more than SimFix . On IntroClassJava and QuixBugs, \techname repairs 35 bugs and 17 bugs respectively, which also achieves better performance than the existing APR tools that were evaluated on the two benchmarks. The results suggest that \techname has better performance and better generalizability than existing approaches. To our knowledge, \emph{this is the first DL-based APR approach that has outperformed traditional APR approaches}. 
}
To summarize, this paper makes the following contributions:
\begin{itemize}
    \item We propose a syntax-guided edit decoder for APR with a provider/decider architecture to accurately predict the edits and ensure that the edited program is syntactically correct and uses placeholders to generate patches with project-specific identifiers.
    \item We design \techname, a neural APR approach based on the decoder architecture described above.
    \item We evaluate \techname on 395 bugs from Defects4J v1.2 and 420 additional bugs from Defects4J v2.0. The results show that \techname significantly outperforms state-of-the-art approaches for single-hunk bugs in terms of both repair performance and generalizability.
\end{itemize}

 

\tcbset{colframe = white,left=0mm,colback=white}
\section{Edits}
We introduce the syntax and semantics of edits and their relations to providers in this section. The neural architecture to generate edits and implement providers will be discussed in the next section.

\subsection {Syntax and Semantics of Edits}

\begin{figure}
\begin{tcolorbox}
$\begin{array}{rlrl}
1.\hspace{-2mm}&{\it Edits} &\rightarrow & {\it Edit}; {\it Edits} \mid {\tt end}\\
2.\hspace{-2mm}&{\it Edit} &\rightarrow & {\it Insert} \mid {\it Modify}\\
3.\hspace{-2mm}&{\it Insert} &\rightarrow & {\tt insert}({\it \langle HLStatement\rangle}) \\
4.\hspace{-2mm}&{\it Modify} &\rightarrow & {\tt modify}( \\
& & &\  \langle\textit{ID of an AST Node with a NTS}\rangle, \\
& & &\  \langle\textit{the same NTS as the above NTS}\rangle) \\
5.\hspace{-2mm}&{\langle \textit{Any NTS in HL}\rangle }\hspace{-2mm}&\rightarrow \\
&\multicolumn{3}{r} {{\tt copy}(\langle \textit{ID of an\ AST\ Node with the same NTS}\rangle) } \\
&&\mid&{\it \langle The\ original\ production\ rules\ in\ HL\rangle} \\
6.\hspace{-2mm}&{\it \langle HLIdentifier\rangle} \hspace{-2mm}&\rightarrow &{\tt placeholder}\\
&&\mid&{\it \langle Identifiers\ in\ the\ training\ set\rangle} \\
\end{array}$
\end{tcolorbox}    
\parbox{\columnwidth}{\footnotesize ``HL'' stands for ``host language''. ``NTS'' stands for ``non-terminal symbol''. ``${\it \langle HLStatement\rangle}$'' is the non-terminal in the grammar of the host language representing a statement. ``${\it \langle HLIdentifier\rangle}$'' is the non-terminal in the grammar of the host language representing an identifier. }
\vspace{-3mm}
\caption{The Syntax of Edits}
    \vspace{-2mm}
    \label{fig:syntax}
\end{figure}

Figure \ref{fig:syntax} shows the syntax of edits. 
Note that our approach is not specific to a particular programming language and can be applied to any programming language (called the \emph{host language}) that has a concept similar to the statement. In particular, it is required that when a statement is present in a program, a sequence of statements can also be present at the same location. In other words, inserting a statement before any existing statement would still result in a syntactically correct program. 
To ensure syntactic correctness of the edited program, the syntax of edits depends on the syntax of the host language. In Figure~\ref{fig:syntax}, ``HL'' refers to the host programming language our approach applies to. \modify{Link rules to figure 4}{In the following we explain each rule in Figure~\ref{fig:syntax} in order.} 

As defined by Rule 1 and Rule 2, an {\it Edits} is a sequence of {\it Edit} ended by a special symbol {\tt end}. An {\it Edit} can be one of two edit operations, {\tt insert} and {\tt modify}.

Rule 3 defines the syntax of {\tt insert} operation. The \texttt{insert} operation inserts a newly generated statement before the faulty statement. As shown in Rule 3, the {\tt insert} operation has one parameter, which is the statement to insert. Here $\langle HLStatement\rangle$ refers to the non-terminal in the grammar of the host language that represents a statement. This non-terminal could be expanded into a full statement, or a copy operation that copies a statement from the original program, or a mixture of both. This behavior will be explained later in Rule 5.

Rule 4 defines the syntax of {\tt modify} operation. The \texttt{modify} operation replaces an AST subtree in the faulty statement with a new AST subtree. The {\tt modify} operation has two parameters. The first parameter is the ID of the root node from the AST subtree to be replaced. The ID of a node is defined as the order of a node in the pre-order traversal sequence, e.g., the 6th visited node has the ID of 6. 
The second parameter is an AST subtree whose root node has the same symbol, i.e., the root node cannot be changed. In this way, the replacement ensures syntactic correctness. To ensure that there is an actual change, the subtree to be replaced should have more than one node, i.e., the root node should have a non-terminal symbol. 

For both {\tt insert} and {\tt modify}, we need to generate a new AST subtree. It is noticeable that in many patches, the AST subtree being inserted or modified is not completely original; some of its subtrees may be copied from other parts of the program. Taking advantage of this property, {\tt copy} operation is introduced to further reduce the patch space. Rule 5 defines the syntax of this operation. It is a meta-rule applied to any non-terminal symbol of the host language. For any non-terminal symbol in the host language, we add a production rule that expands it into a {\tt copy} operation. The original production rules for this non-terminal are also kept, so that when generating the edits, the neural network could choose to directly generate a new subtree or to copy one.

The {\tt copy} operation has one parameter, which identifies the root node of the AST subtree to be copied. The AST subtree can be selected from the faulty statement or its context. 
In our current implementation, we allow copying from the method surrounding the faulty statement. Also, to ensure syntactic correctness, the root node of the subtree to be copied should have the same non-terminal symbol as the symbol being extended. 


Finally, Rule 6 introduces {\tt placeholder} into the grammar. Normally, the grammar of a programming language uses a terminal symbol to represent an identifier. To enable the neural network to generate concrete identifiers as well as the {\tt placeholder}, we change identifier nodes into non-terminals, which expand to either {\tt placeholder} or one of the frequent identifiers in the training set. In our current implementation, an identifier is considered frequent if it appears more than 100 times in the training set. 

When applying the edits, the {\tt placeholder} tokens are replaced with feasible identifiers within the context. We first collect all identifiers in the current projects by performing a lexical analysis and collect the tokens whose lexical type is $\tt \langle HLIdentifier\rangle$, the symbol representing an identifier in the host language. Then we filter identifiers based on the following criteria: (1) the identifier is accessible from the local context, and (2) replacing the placeholder with the identifier would not lead to type errors. The remaining identifiers are feasible identifiers.

Figure~\ref{fig:insert} and Figure~\ref{fig:modify} show two example patches represented by edits. The patch in Figure~\ref{fig:insert} inserts an {\tt if} statement, and the conditional expression contains a method invocation that is copied from the faulty statement. The patch in Figure~\ref{fig:modify} replaces the qualifier of a method invocation with another invocation, where the name of the method is a placeholder to be instantiated later.

\begin{theorem}
    The edited programs are syntactically correct. 
\end{theorem}
\begin{proof}
It is easy to see that the theorem holds by structural induction on the grammar of the edits. First, the requirement on the host programming language ensures that inserting a statement before another statement is syntactical correct. Second, when replacing a subtree with {\tt modify}, the root symbol of the subtree remains unchanged. Third, the new subtree in {\tt insert} and {\tt modify} is generated by either using the grammar rules of the host language, or copying a subtree with the same root symbol. Finally, instantiating a {\tt placeholder} ensures syntactic correctness because we only replace a {\tt placeholder} with a token whose lexical type is $\tt \langle HLIdentifier\rangle$.
\end{proof}


\begin{figure}
    \centering
    \includegraphics[width=\linewidth]{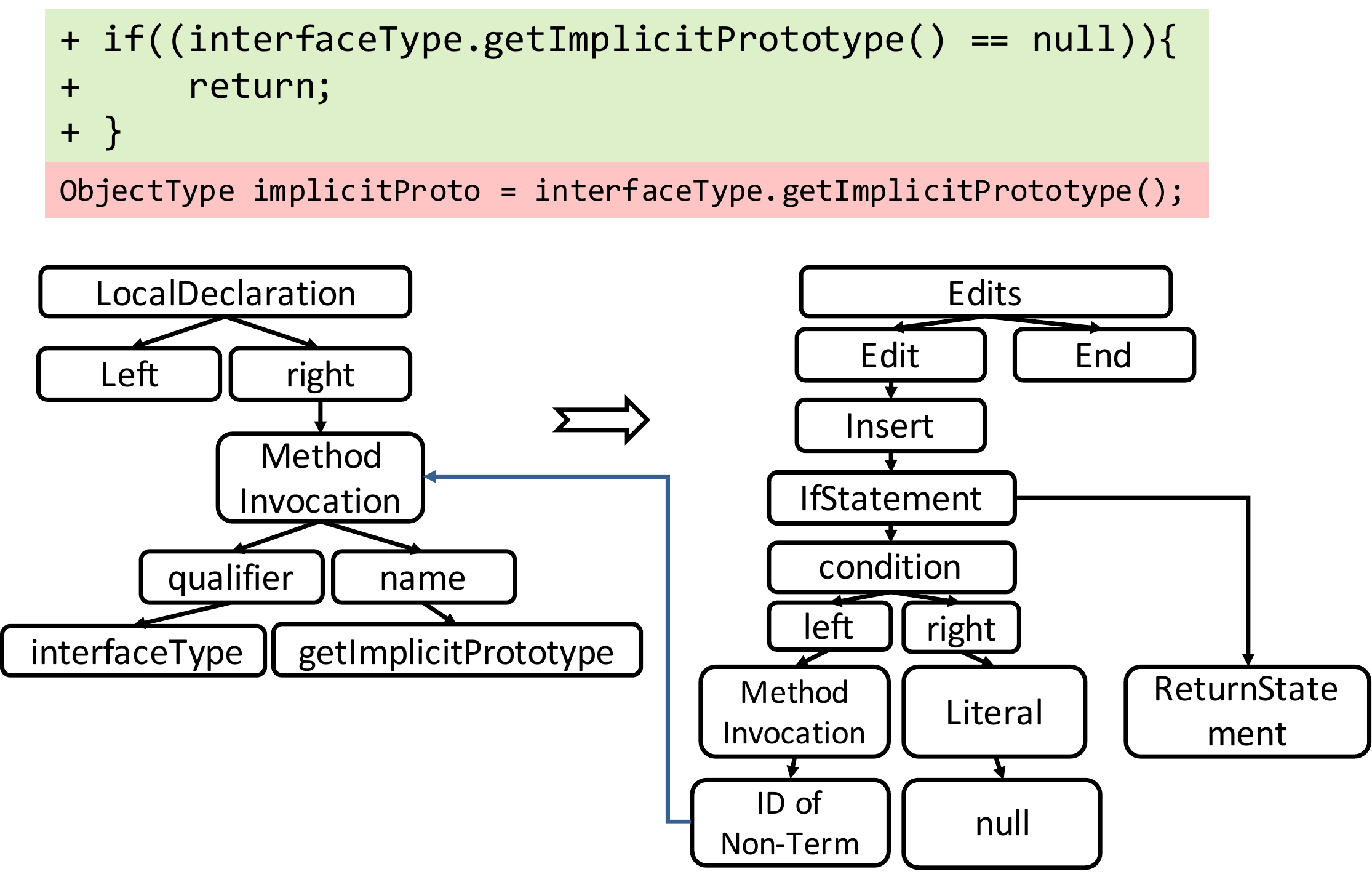}
    \vspace{-6mm}
    \caption{Example of Insert Operation (Closure-2)}
    \vspace{-2mm}
    \label{fig:insert}
\end{figure}

\begin{figure}
    \centering
    \includegraphics[width=\linewidth]{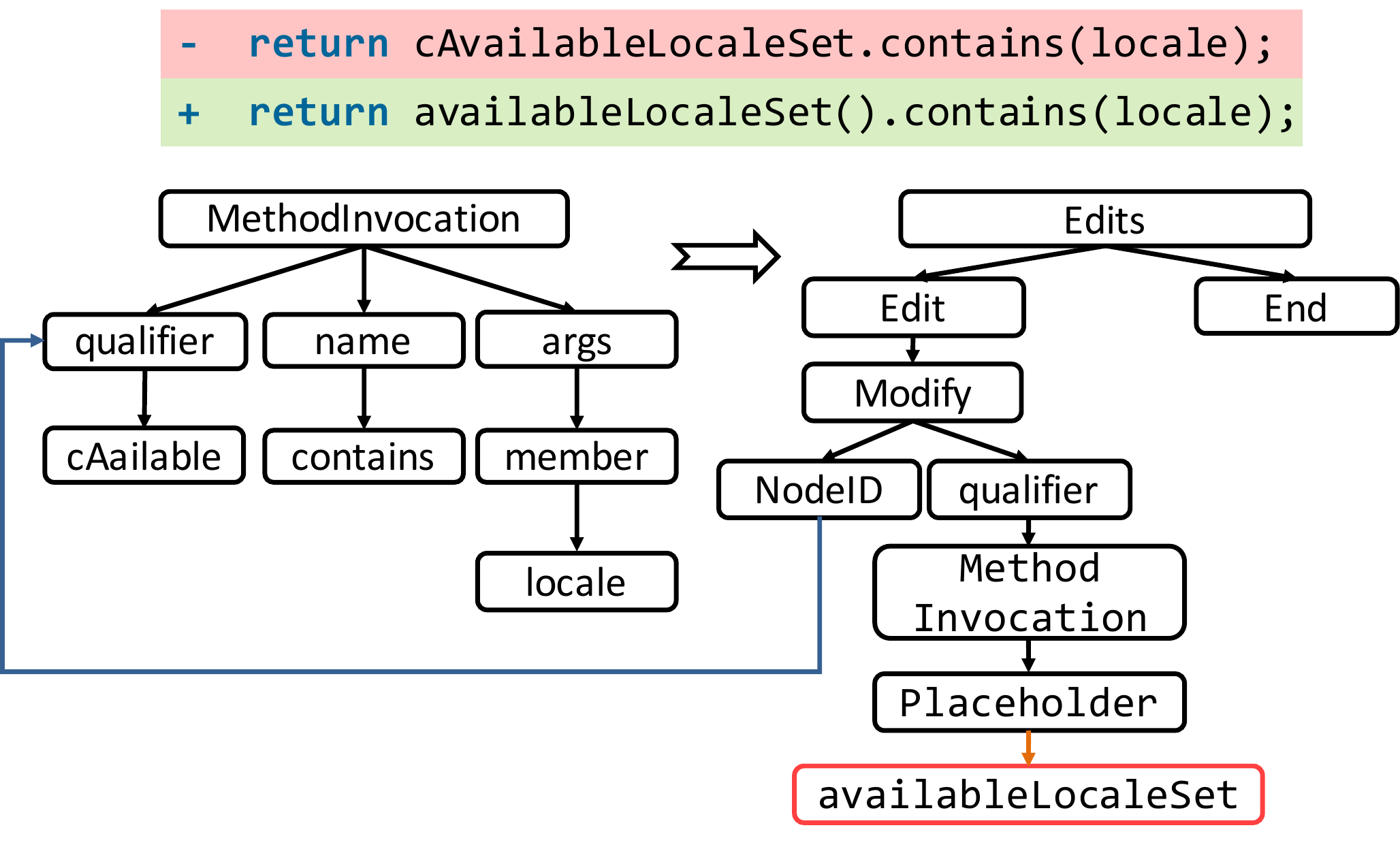}                            \vspace{-8mm}

    \caption{Example of Modify Operation (Lang-57)}
    \label{fig:modify}
\end{figure}
\begin{figure*}
    \centering
    \includegraphics[width=\linewidth]{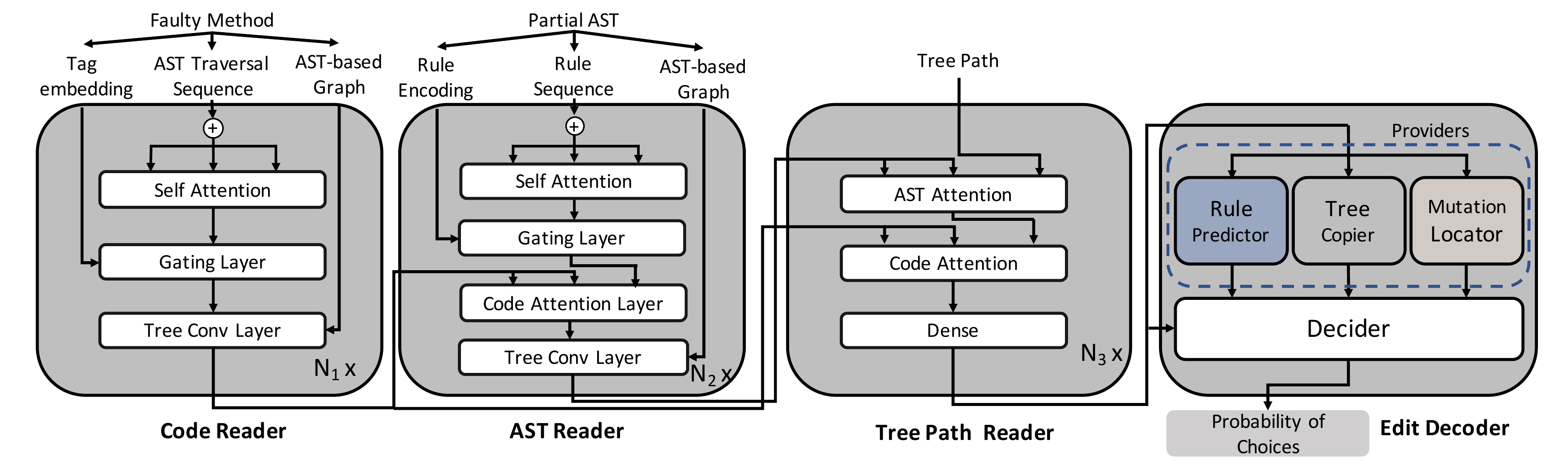}
                            \vspace{-9mm}
    \caption{\modify{Overview of Model}{Overview of \techname}}\vspace{-2mm}
    \label{fig:overview-model}
\end{figure*}
\subsection{Generation of Edits}

\begin{table}[!t]
    \caption{Providers for non-terminals}
    \vspace{-2mm}
    \centering
    \small
    \begin{tabular}{l|l}
        \toprule
        {\textbf{ Component}} & \textbf{Associated Non-terminals}   \\
        \midrule
        Rule Predictor & \it Edits, Edit, Insert, $\langle \it HLIdentifier \rangle$, $\langle \textit{Any NTS in HL} \rangle$ \\
        \midrule
        Subtree Locator & \it Modify  \\
        \midrule
        Tree Copier & $\langle \textit{Any NTS in HL} \rangle$  \\
        \bottomrule
    \end{tabular}
    \label{tab:type}
\end{table}

Since the choice of expanding a non-terminal may depend on the local context or a previous choice, we use providers to provide choices and estimate their probabilities. Our current implementation has three types of providers. Table~\ref{tab:type} shows these providers and their associated non-terminals.

For non-terminals {\it Edits, Edit, Insert} and $\langle \it HLIdentifier \rangle$, the rule predictor is responsible for providing choices and estimates the probability of each production rule. The rule predictor consists of a neural component and a logic component. After the neural component assigns the probability for each production rule, the logic component resets the probability of rules whose left-hand side is not the corresponding non-terminal to zero and normalizes the remaining probabilities.

For {\it Modify}, the subtree locator is responsible for providing the choices. The subtree locator estimates the probability of each AST subtree with a size larger than 1 in the faulty statement. The choice of a subtree $t$ means that we should expand {\it Modify} into {\tt modify({\it ID}, {\it NTS})} where {\it ID} is the root ID of $t$ and {\it NTS} is the root symbol of $t$.

For any non-terminal in the grammar of the host language (note that $\langle \it HLIdentifier \rangle$ is a terminal symbol in the host language), both the rule predictor and the tree copier are responsible to provide the choices. The tree copier estimates the probabilities of each AST subtree with a size larger than 1 in the method surrounding the faulty statement. The choice of a subtree $t$ means that we should expand the non-terminal into {\tt copy({\it ID})}, where {\it ID} is the root ID of $t$. Similar to the rule predictor, the tree copier employs a logic component after the neural component to reset the probabilities of subtrees whose root symbols are different from the non-terminal symbol being expanded.

Finally, the decider assigns a probability to each provider. The decider also includes a similar logic component, which resets the probability of a provider to zero if that provider is not responsible for the current non-terminal symbol. For example, if the symbol being expanded is \textit{Modify}, the decider resets the probability of rule predictor and tree copier to zero.

\vspace{-3mm}
\section{Model Architecture}
The design of our model is based on the state-of-the-art syntax-guided code generation model, TreeGen~\cite{treegen}. It is a tree-based Transformer~\cite{transformer} that takes a natural language description as input and produces a program as output. Since our approach takes a faulty statement and its context as input and produces edits as output, we replace the components in TreeGen for encoding natural language description and decoding the program.

Figure~\ref{fig:overview-model} shows an overview of our model. The model performs one step in the edit generation process, which is to predict probabilities of choices for expanding a non-terminal node. Beam search is used to find the best combination of choices for generating the complete edits. The model consists of four main components: 
\begin{itemize}
    \item The {\bf code reader} that encodes the faulty statement and its context.
    \item The {\bf AST reader} that encodes the partial AST of the edits that have been generated.
    \item The {\bf tree path reader} that encodes a path from the root node to a non-terminal node which should be expanded.
    \item The {\bf edit decoder} that takes the encoded information from the previous three components and produces a probability of each choice for expanding the non-terminal node.
\end{itemize}
Among them, the AST reader and the tree path reader are derived from TreeGen, where the code reader and the edit decoder are newly introduced in this paper. In this section, we focus on describing the latter two components in detail. 

\subsection{Code Reader}
The code reader component encodes the faulty statement and the method surrounding the faulty statement as its context, \modify{}{where the faulty statement is localized by a fault localization technique}. It uses the following three inputs. \textbf{(1) AST traversal sequence.} This is a sequence of tokens following the pre-order traversal of the AST, $\bm{c}_1, \bm{c}_2, \cdots{}, \bm{c}_L$, where $\bm{c}_i$ is the token encoding vector of the $i$th node embedded via word embedding~\cite{Mikolov2013EfficientEO}. \textbf{(2) Tag embedding.} This is a sequence of tags following the same pre-order traversal of the AST, where each tag denotes which of the following cases the corresponding node belongs to: 1. in the faulty statement, 2. in the statement before the faulty statement, 3. in the statement after the faulty statement, or 4. in other statements. Each tag is embedded via an embedding-lookup table. We denote the tag embedding as $\bm{t}_1,\bm{t}_2,\cdots{},\bm{t}_L$. \textbf{(3) AST-based Graph.} Considering that the former two inputs do not capture the neighbor relations between AST nodes, in order to capture such information, we treat an AST as a directional graph where the nodes are AST nodes and the edges link a node to each of its children and its left sibling, as shown in Figure~\ref{fig:ast}(b). This graph is embedded as an adjacent matrix.



The code reader uses three sub-layers to encode the three inputs above, as discussed in the following sections.

\subsubsection{Self-Attention}
The self-attention sub-layer encodes the AST traversal sequence, following the Transformer~\cite{transformer} architecture to capture the long dependency information in the AST. 

Given the embedding of the input AST traversal sequence, we use position embedding to represent positional information of the AST token. The input vectors are denoted as $\bm{c}_1, \bm{c}_2, \cdots{}, \bm{c}_L$, and the position embedding of $i$th token is computed as
        \begin{align}
        \vspace{-2mm}
                 p_{(i,2j)} &= 
                    \sin (pos/(10000^{2j/d} )) \\ 
                     p_{(i,2j + 1)} &= \cos (pos/(10000^{2j/d} )) 
        \vspace{-3mm}
        \end{align}
where $pos = i + step$, $j$ denotes the element of the input vector and $step$ denotes the embedding size. After we get the vector of each position, it is directly added to the corresponding input vector, where $\bm{e}_i = \bm{c}_i + \bm{p}_i$.

Then, we adopt multi-head attention layer to capture non-linear features. Following the definition of \citet{transformer}, we divide the attention mechanism into $H$ heads. Each head represents an individual attention layer to extract unique information. The single attention layer maps the query $Q$, the key $K$, and the value $V$ into a weighted-sum output. The computation of the $j$th head layer can be represented as
\begin{equation}
    \vspace{-2mm}
    head_j = \text{softmax}(\frac{QK^T}{\sqrt{d_k}})V\label{eq:att2}
    \vspace{-1mm}
\end{equation}
where $d_k = d/H$ denotes the length of each extracted feature vector, and $Q$, $K$ and $V$ are computed by a fully-connected layer from $Q$, $K$, $V$. In the encoder, vectors $Q$, $K$ and $V$ are all the outputs of the position embedding layer $\bm{e}_1, \bm{e}_2, \cdots, \bm{e}_L$. The outputs of these heads are further joint together with a fully-connected layer, which is computed by
\begin{equation}
    \vspace{-1mm}
    Out = [head_1;\cdots;head_H]\cdot W_h \label{eq:att3}
    \vspace{-1mm}
\end{equation}
where $W_h$ denotes the weight of the fully-connected layer and $Out$ denotes the outputs $\bm{a_1}, \bm{a_2}, \cdots{}, \bm{a_L}$ of the self-attention sub-layer. 

\subsubsection{{Gating Layer}}
This sub-layer takes the outputs of the previous layer and the tag embedding as input. 
Gating mechanism, as defined in TreeGen~\cite{treegen}, is used in this layer. It takes three vectors named $\bm{q}, \bm{c}_1, \bm{c}_2$ as input and aims to corporate $c_1$ with $c_2$ based on $\bm{q}$. The computation of gating mechanism can be represented as 
\begin{align}
    \alpha_i^{c_1} = \exp(\bm{q}_i^T\bm{k}_i^{c_1})&/\sqrt{d_k}\\
\alpha_i^{c_2} = \exp(\bm{q}_i^T\bm{k}_i^{c_2})&/\sqrt{d_k}\\
    \bm{h}_i = (\alpha_i^{c_1}\bm{v}_i^{c_1} + \alpha_i^{c_2}\bm{v}_i^{c_2})&/(\alpha_i^{c_1} + \alpha_i^{c_2})
    \label{eq:gate}
\end{align}
where $d_k = d / H$ is a normalization factor, $H$ denotes the number of heads, and $d$ denotes the hidden size; $\bm{q}_i$ is computed by a fully-connected layer over the control vector $\bm{q}_i$; $\bm{k}_i^{c_1}, \bm{v}_i^{c_1}$ is computed by another fully-connected layer over vector $\bm{c}_1$; $\bm{k}_i^{c_2}$ and $\bm{v}_i^{c_2}$ are also computed by the same layer with different parameters over the vector $\bm{c}_2$. 

In our model, we treat the outputs of the self-attention sub-layer $\bm{a}_1, \bm{a}_2, \cdots{}, \bm{a}_L$ as $\bm{q}$ and $\bm{c}_1$, and the tag embedding $\bm{t}_1, \bm{t}_2, \cdots{}, \bm{t}_L$ as $\bm{c}_2$. Thus, embedding of the $i$th AST node of the gating-layer can be represented as $\bm{u}_i = \text{Gating}(\bm{a_i}, \bm{a_i}, \bm{t}_i)$. 
 
 \begin{figure}
    \centering
    \includegraphics[width=.95\linewidth]{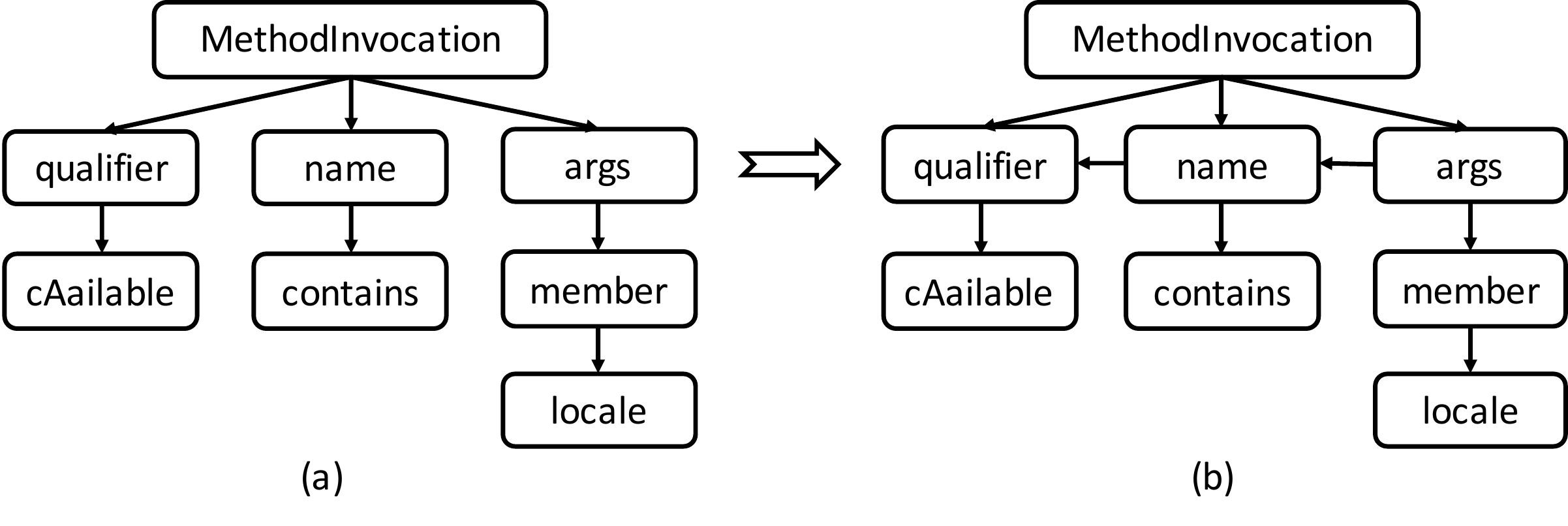}
                            \vspace{-5mm}
    \caption{Example of AST-based Graph}
    \vspace{-2mm}
    \label{fig:ast}
\end{figure}

\subsubsection{{Tree Conv Layer}}
This sub-layer takes the output $\bm{u}_i$ of the previous layer and the AST-based graph $G$ (represented as an adjacency matrix) as input.
We adopt a GNN~\cite{gnn08, Zhang_2020} layer to process the inputs, and the encoding of the neighbors $r_i$ is computed as 
\begin{equation}
    \bm{g}_i = W_g \sum_{r^j \in G}A^{n}_{r^i r^j}\bm{u}^j
    \vspace{-1mm}
\end{equation}
where $W_g$ is the weight of a fully-connected layer and $\hat{A}$ is a normalized adjacency matrix of $G$. The computation of the normal operation proposed by \citeauthor{kipf2016semi}~\cite{kipf2016semi} is represented as $\hat{A} = S_1^{-1/2} A S_2^{-1/2}$,
where $A$ is the adjacency matrix of $G$, and $S_1, S_2$ are the diagonal matrices with a summation of $A$ in columns and rows. Then, the encoding of the neighbors is directly added to the input vector.

In summary, the code reader has $N_1$ blocks of these three sub-layers, and yields the features of the input AST, $\bm{t_1}, \bm{t_2},\cdots,\bm{t_L}$, which would be used for the AST reader and the tree path reader.
\subsection{AST Reader}
\modify{Description of TreeGen}{
The AST reader encodes the partial generated AST of the edit, which has the same structure as the one in TreeGen~\cite{treegen}. This component takes three inputs derived from the partial generated AST as code reader. The rule sequence is represented as real-value vectors and then is fed into a self-attention layer. We then integrate the output of the self-attention layer with the rule encoding via a gating layer as Equation~\ref{eq:gate}. We also adopt a multi-head attention layer over the outputs of the code reader and the gating layer like the decoder-encoder attention in Transformer. Finally, we use a tree convolutional layer like the code reader to extract the structural information. More details of this component can be found in the publication of TreeGen~\cite{treegen}.}
\subsection{Tree Path Reader}
\modify{Description of TreeGen}{
The tree path reader encodes the information of the non-terminal node to be expanded and is the same as the one in TreeGen~\cite{treegen}. This component represents the non-terminal node as a path from the root to the node to be expanded and transforms the nodes in this path into real-value vectors. As shown in Figure~\ref{fig:overview-model}, these vectors are fed into two attention layers like Equation~\ref{eq:att3}. Finally, a set of two fully-connected layers, where the first layer has a $GELU$~\cite{2016Bridging} activation function, are followed to extract features for edit decoder. More details of this component can be found in the publication of TreeGen~\cite{treegen}. We denote the output of the tree path reader as $\bm{d}_1, \bm{d}_2,\cdots{}, \bm{d}_T$.}
\subsection{Edit Decoder}
The edit decoder takes the output of the tree path reader, $\bm{d}_1, \bm{d}_2,\cdots{}, \bm{d}_T$ with length $T$, as input. These vectors are produced by the tree path reader and contain the encoded information from all the inputs: the faulty statement with its surrounding method, the partial AST generated so far, and the tree path denoting the node to be expanded.

\subsubsection{Provider}
As mentioned before, there are currently three types of providers: rule predictor, tree copier, and subtree locator. These providers take the vector $\bm{d}_1, \bm{d}_2,\cdots{}, \bm{d}_T$ as input and output the probability of choices for different non-terminals.
\paragraph{\textbf{Rule Predictor}}
 The rule predictor estimates the probability of each production rule in the grammar of edits. The neural component of this decider consists of a fully-connected layer. The output of the fully-connected layer is denoted as $\bm{s}_1, \bm{s}_2, \cdots{}, \bm{s}_T$. Then, these vectors are normalized via softmax, which computes the normalized vectors $\bm p_1^r, \bm p_2^r,\cdots{}, \bm p_T^r$ by
\begin{equation}
    \bm p^r_k(m) = \frac{\exp\{\bm{s}^m_k\}}{\sum_{j=1}^{N_r}\exp\{\bm{s}_k^j\}}
    \label{con:soft}
\end{equation}
where $N_r$ denotes the number of production rules in the grammar of edits, and $m$ denotes the $m$th dimension of the vector $\bm p^r_k$ (i.e., the production rule with ID $m$). 
In particular, invalid rules whose left-hand side is not the corresponding non-terminal are not allowed in our approach. For these rules, the logic component resets the output of the fully-connected layer to $-\infty$. Thus, the probability of invalid rules will be zero after softmax normalization. 


\paragraph{\textbf{Tree Copier}}

This provider is designed for any non-terminal symbol in the grammar of edits to choose a subtree in the local context. The neural component is based on a pointer network~\cite{pointer}. The computation can be represented as 
\begin{equation}
    \bm{\theta}_i = \bm{v}^T\text{tanh}({W}_1\bm{d}_i + {W}_2\bm{t})
    \label{con:1}
\end{equation}
where $\bm{t}$ denotes the output of the code reader, and $\bm{v}, {W}_1, {W}_2$ denote the trainable parameters. The logic component also resets $\bm{\theta}$ to $-\infty$ if the root symbol of the corresponding subtree is different from the symbol being expanded. These vectors are then normalized via softmax as Equation~\ref{con:1}. We denote the normalized vector as $\bm p_1^t, \bm p_2^t,\cdots{}, \bm p_T^t$.
\paragraph{\textbf{Subtree Locator}}
This component outputs an ID of the subtree in the faulty statement for not-terminal symbol, \textit{Modify}, in the grammar of edits. The computation of this component is the same as the tree copier. We denote the output vector of this provider as $\bm p_1^s, \bm p_2^s,\cdots{}, \bm p_T^s$ 

\subsubsection{Decider}
For these three providers, the decider estimates the probability of using each provider. The neural component also takes the output of the tree path reader, $\bm{d}_1, \bm{d}_2,\cdots{}, \bm{d}_T$, as input, and produces the probability of using each provider as output. The computation can be represented as $\bm \lambda_i = {W}\bm{d}_i + \bm{b}$, where ${W}$ and $\bm{b}$ denote the parameters of a fully-connected layer. The logic component resets $\bm \lambda$ to $-\infty$ if the corresponding provider is not responsible for the symbol being expanded following Table~\ref{tab:type}. Then, the vectors are normalized via softmax as Equation \ref{con:soft}. We denote the normalized vectors as $\bm \lambda_1, \bm \lambda_2,\cdots{},\bm \lambda_T$. The final probability of each choice can be computed as 
\begin{equation}
    \bm o_i = [\bm \lambda_i^r \bm p_i^r;\bm \lambda_i^t \bm p_i^t;\bm \lambda_i^s \bm p_i^s]
\end{equation}
where $\bm o_i$ will be the probability vector of the next production rule at $i$th step during patch generation.

\subsection{Training and Inference}
During training, the model is optimized by maximizing the negative log-likelihood of the oracle edit sequence and \modify{Explanation}{we do not use the logic component in the providers and decider. Here we would like Recoder to learn the distribution of the rules handled by the logic component. If the logic component is present at training, Recoder would not be trained for a large portion of rules. During inference, these unseen rules would distort the distribution of output, making Recoder fail to distinguish the part of rules that it is supposed to distinguish.}

When generating edits, inference starts with the rule $start: \textit{start} \longrightarrow \textit{Edits}$, expanding a special symbol \texttt{start} to \texttt{Edits}. The recursive prediction terminates if every leaf node in the predicted AST is a terminal. We use beam search with a size of 100 to generate multiple edits. 

Generated edits may contain placeholders. Though the number of choices for a single placeholder is small, the combination of multiple placeholders may be large. Therefore, we discard patches containing more than one \texttt{placeholder} symbol during beam search.

\subsection{Patch Generation and Validation}
Patches are generated according to the result of the fault localization technique.
In our approach, the model described above is invoked for each suspicious faulty statement according to the result of fault localization. For each statement, we generate 100 valid patch candidates via beam search: when beam search generates a valid patch, we remove it from the search set and continue to search for the next patch until 100 candidates are generated in total for that statement. 
After patches are generated, the final step is to validate them via the test suite written by developers. The validation step filters out patches that do not compile or fail a test case. 
All generated patches are validated until a plausible patch (a patch that passes all test cases) is found.

\section{Experiment Setup}
\modify{}{We have implemented \techname for the Java programming language. In this and the next sections we report our experiments on repairing Java bugs.}

\subsection{Research Questions}
Our evaluation aims to answer the following research questions:

\noindent
    \textbf{RQ1:} \textbf{What is the performance of \techname?}

To answer this question, we evaluated our approach on the widely used APR benchmark, \textit{Defects4J v1.2}, and compared it with traditional and DL-based APR tools.

\noindent
    \textbf{RQ2:} \textbf{What is the contribution of each component in \techname?}
 
 To answer this question, we started from the full model of \techname, and removed each component in turn to understand its contribution to performance.

\noindent
\textbf{RQ3:} \textbf{What is the generalizability of \techname?}
 
 To answer this question, we fisrt conducted an experiment on 420 additional bugs from \textit{Defects4J v2.0}. To our best knowledge, this is the first APR approach that has been applied to this benchmark. We compared \techname with the previous two best-performing APR approaches for single-hunk bugs on Defects4J v1.2, namely TBar~\cite{liu2019tbar} and SimFix~\cite{2018Shaping}. In addition, we also applied \techname to other two benchmarks, \textit{QuixBugs} and \textit{IntroClassJava}, via \textit{RepairThemAll}~\cite{RepairThemAll2019} framework, which allows the execution of automatic program repair tools on benchmarks of bugs.
 
    
    


\subsection{Dataset}
The neural network model in our approach needs to be trained with a large number of history patches. To create this training set, we crawled Java projects created on GitHub~\cite{github} between March 2011 and March 2018, and downloaded 1,083,185 commits where the commit message contains at least one word from the following two groups, respectively: (1) \textit{fix, solve}; (2) \textit{bug, issue, problem, error}. Commits were filtered to include only patches that modify one single statement or insert one new statement, corresponding to two types of edits that our approach currently supports. To avoid data leak, we further discarded patches where (1) the project is a clone to Defects4J project or a program repair project using Defects4J, or (2) the method modified by the patch is the same as the method modified by any patch in Defects4J v1.2 or v2.0, based on AST comparison. There are 103,585 valid patches left after filtering, 
which are further split into two parts: 80\% for training and 20\% for validation.

We used four benchmarks to measure the performance of \techname. The first one contains 395 bugs from \textit{Defects4J} v1.2~\cite{defects4j}, which is a commonly used benchmark for automatic program repair research. The second one contains 420 additional bugs from \textit{Defects4J v2.0}~\cite{defects4j}. 
Defects4J v2.0 introduces 438 new bugs compared with Defects4J v1.2. However, GZoltar~\cite{GZoltar}, the fault localization approach used by our implementation as well as two baselines (TBar and SimFix), failed to finish on the project {\it Gson}, so we excluded 18 bugs in {\it Gson} from our benchmark. \modify{}{The third one contains 40 bugs from \textit{QuixBugs}~\cite{quixbugs}, which is a benchmark
with 40 buggy algorithmic programs specified by test cases. The last one, \textit{IntroClassJava}~\cite{durieux:hal-01272126}, consists of 297 buggy Java programs generated from the \textit{IntroClass}~\cite{LeGoues15tse} benchmark for C.}

\subsection{Fault Localization}
In our experiment, two settings for fault localization are used. In the first setting, the faulty location of a bug is unknown to APR tools, and they rely on existing fault localization approaches to localize the bug. \techname uses Ochiai~\cite{ochiai} (implemented in GZoltar~\cite{GZoltar}), which is widely used in existing APR tools~\cite{liu2019tbar, 2018Shaping}. In the second setting, the actual faulty location is given to APR tools. This is to measure the capability of patch generation without the influence of a specific fault localization tool, as suggested and adopted in previous studies~\citeN{coconut, codit, 9000077}.


\subsection{Baselines}
We selected existing APR approaches as the baselines for comparison. Since \techname generates only single-hunk patches (patches that only change a consecutive code fragment), we chose 10 traditional single-hunk APR approaches that are often used as baselines in existing studies: {jGenProg}~\cite{genprog}, {HDRepair}~\cite{le2016history}, {Nopol}~\cite{xuan2016nopol}, {CapGen}~\cite{capgen}, {SketchFix}~\cite{hua2018sketchfix}, {TBar}~\cite{liu2019tbar},
{FixMiner}~\cite{koyuncu2020fixminer}, {SimFix}~\cite{2018Shaping}, {PraPR}~\cite{prapr}, {AVATAR}~\cite{liu2019avatar}. In particular, TBar correctly repairs the highest number of bugs on Defects4J v1.2 as far as we know. We also selected DL-based APR approaches that adopt the encoder-decoder architecture to generate patches and have been evaluated on Defects4J as baselines. Four approaches have been chosen based on this criteria, namely, {SequenceR}~\cite{9000077}, {CODIT}~\cite{codit}, {DLFix}~\cite{dlfix}, and {CoCoNuT}~\cite{coconut}. 

For Defects4J v1.2, the performance data of the baselines are collected from existing papers~\cite{liu2019tbar, liu2020efficiency}.
For additional bugs from Defects4J v2.0, two best-performing single-hunk APR approaches on Defects4J v1.2, TBar and SimFix, are adapted and executed for comparison. \modify{}{For QuixBugs and IntroClassJava, we directly choosed the APR tools used in RepairThemAll~\cite{RepairThemAll2019} and DL-based APR tools which have experimented on these two benchmarks as baselines: jGenProg~\cite{genprog}, RSRepair~\cite{PatchPlausibility}, Nopol~\cite{xuan2016nopol}, and CoCoNuT~\cite{coconut}. We also directly used the result reported in the original papers~\cite{RepairThemAll2019, coconut}}.

\subsection{Correctness of Patches}
To check the correctness of the patches, we manually examined every patch if it is the same with or semantically equivalent to the patch provided by Defects4J, as in previous works~\cite{liu2019tbar, 2018Shaping, dlfix, coconut, prapr}. To reduce possible errors made in this process, every patch is examined by two of the authors individually and is considered correct only if both authors consider it correct. The kappa score of the experiment is 0.98.
Furthermore, we also publish all the patches generated by \techname for public judgment\footnote{The source code of \techname, generated patches, and an online demo are available at https://github.com/pkuzqh/Recoder}.

\subsection{Implementation Details}
Our approach is implemented based on PyTorch~\cite{pytorch}, with parameters set to $N_1 = 5, N_2 = 9, N_3 = 2$, i.e., the code reader contains a stack of 5 blocks, the AST reader contains a stack of 9 blocks, and the decoder contains a stack of 2 blocks, respectively. Embedding sizes for all embedding vectors are set to 256, and all hidden sizes are set following the configuration of TreeGen~\cite{treegen}. During training, dropout~\cite{dropout} is used to prevent overfitting, with the drop rate of 0.1. The model is optimized by Adam~\cite{adam} with learning rate 0.0001. These hyper-parameters and parameters for our model are chosen based on the performance on validation set. 

We set a 5-hour running-time limit for \techname, following existing studies~\cite{2018Shaping, dlfix, coconut,saha2019harnessing}. 

\begin{table*}
\caption{Comparison without Perfect Fault Localization}                            \vspace{-2mm}

\begin{threeparttable}
  \begin{tabular}{l|c|c|c|c|c|c|c|c|c|
  c|c|c}
    \toprule
    Project&jGenProg&HDRepair&Nopol&CapGen&SketchFix&FixMiner&SimFix&TBar&DLFix&PraPR&AVATAR&\techname\\
    \midrule
    Chart & 0/7 & 0/2 & 1/6 & 4/4 & 6/8&5/8&4/8&\textbf{9/14}&5/12&4/14&5/12&8/14\\
    Closure & 0/0 & 0/7 & 0/0 & 0/0 & 3/5 & 5/5 & 6/8 & 8/12 & 6/10 &12/62&8/12& \textbf{15/31}\\
    Lang & 0/0 & 2/6 & 3/7 & 5/5 & 3/4 & 2/3 & \textbf{9/13}&5/14&5/12& 3/19 &{5/11}&\textbf{9/15}\\
    Math & 5/18 & 4/7 & 1/21 & 12/16 & 7/8 & 12/14 & 14/26 & \textbf{18/36} & 12/28 &6/40&{6/13}& {15/30}\\
    Time & 0/2 & 0/1 & 0/1 & 0/0 & 0/1 & 1/1 & 1/1 & 1/3 & 1/2&0/7& {1/3}&\textbf{2/2}\\
    Mockito & 0/0 & 0/0 & 0/0 & 0/0& 0/0 & 0/0 & 0/0&1/2 & 1/1 &1/6&\textbf{2/2} & \textbf{2/2}\\
    \midrule
    Total & 5/27 & 6/23 & 5/35 & 21/25 & 19/26& 25/31 & 34/56 & 42/81 & 30/65&26/148&27/53 & \textbf{51/94}\\
    \midrule
    P(\%) & 18.5 & 26.1 & 14.3 & \textbf{84.0} & 73.1 & 80.6 & 60.7 & 51.9 & 46.2 &17.6&50.9& 54.3\\
  \bottomrule  
\end{tabular}
 \begin{tablenotes}
 \footnotesize
 \item In the cells, x/y:x denotes the number of correct patches, and y denotes the number of patches that can pass all the test cases.
 \end{tablenotes}
\end{threeparttable}
\label{tab:result}
\end{table*}

\begin{table}
\caption{Comparison with Perfect Fault Localization}\vspace{-2mm}
  \resizebox{\linewidth}{!}{
  \begin{tabular}{l|c|c|c|c|c|c}
    \toprule
    Project&SequenceR&CODIT&DLFix&CoCoNuT&TBar&\techname\\
    \midrule
    Chart & 3 & 4 & 5 & 7&\textbf{11} & {10}\\
    Closure & 3 & 3 & 11 & 9&17 & \textbf{23}\\
    Lang & 3 & 3 & 8 & 7&\textbf{13} & 10\\
    Math & 4 & 6 & 13 & 16&\textbf{22} & 18\\
    Time & 0 & 0 & 2 & 1& 2& \textbf{3} \\
    Mockito & 0 & 0 & 1& \textbf{4} & 3 & {2}\\
    \midrule
        Total & 13 & 16 & 40 & 44 & \textbf{68}&{66}\\
  \bottomrule  
\end{tabular}
}
\label{tab:result2}
\end{table}

\section{Experimental Results}
\subsection{Performance of \techname(RQ1)}
\subsubsection{Results without Perfect Fault Localization}
We first compare \techname with the baselines in the setting where no faulty location is given. Results as Table~\ref{tab:result} shown only include baselines that have been evaluated under this setting. As shown, \techname correctly repairs 51 bugs and outperforms all of the previous single-hunk APR techniques on Defects4J v1.2. In particular, \techname repairs 21.4\% (9 bugs) more bugs than the previous state-of-the-art APR tool for single-hunk bugs, TBar. Within our knowledge, \techname is the first DL-based APR approach that has outperformed the traditional APR approaches. 

We show a few example patches that are possibly generated with the help of the novel techniques in \techname.
As shown in Figure~\ref{fig:unique2}, Chart-8 is a bug that DLFix fails to fix. The correct patch only changes a parameter of the method invocation while DLFix needs to generate the whole expression. 
By contrast, \techname generates a {\it modify} operation that changes only one parameter.
Figure~\ref{fig:unique} shows a bug only repaired by \techname. This patch relies on a project-specific method, ``\texttt{isNoType}'', and thus cannot be generated by many of the existing approaches. However, \techname  
fixes it correctly by generating a placeholder and then instantiating it with ``\texttt{isNoType}''. 
\begin{figure}
    \centering
    \includegraphics[width=\linewidth]{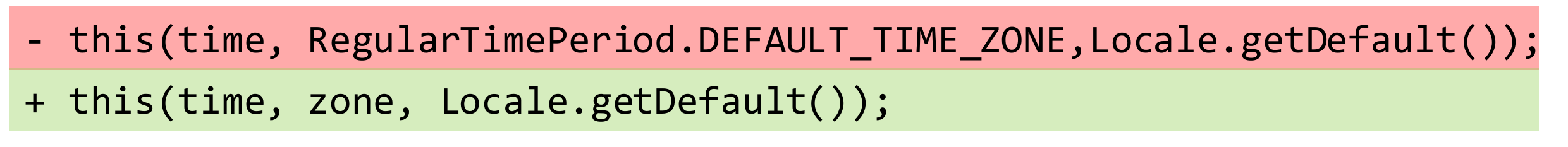}
                            \vspace{-7mm}
    \caption{Chart-8 - A bug fixed by \techname with Modify operation}
    \vspace{-2mm}
    \label{fig:unique2}
\end{figure}
\begin{figure}
  \centering
  \includegraphics[width=\linewidth]{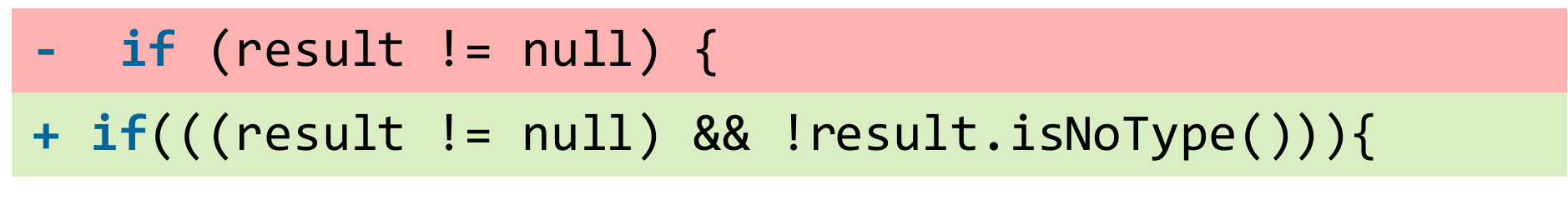}
                            \vspace{-7mm}      

  \caption{Closure-104 - A bug fixed by \techname with \textit{placeholder} generation}
    \vspace{-2mm}
  \label{fig:unique}
\end{figure}
\vspace{-2mm}
\subsubsection{Results with Perfect Fault Localization}
Table~\ref{tab:result2} shows the result where the actual faulty location is provided. 
As before, only baselines that have been evaluated under this setting are listed.
\techname still outperforms all of the existing APR approaches, including traditional ones. 
Also, compared with \techname using Ochiai for fault localization, this model achieves a 35.3\% improvement. The result implies that \techname can achieve better performance with better fault localization techniques.

\subsubsection{Degree of Complementary}
\begin{figure}
    \centering
    \includegraphics[width=\linewidth]{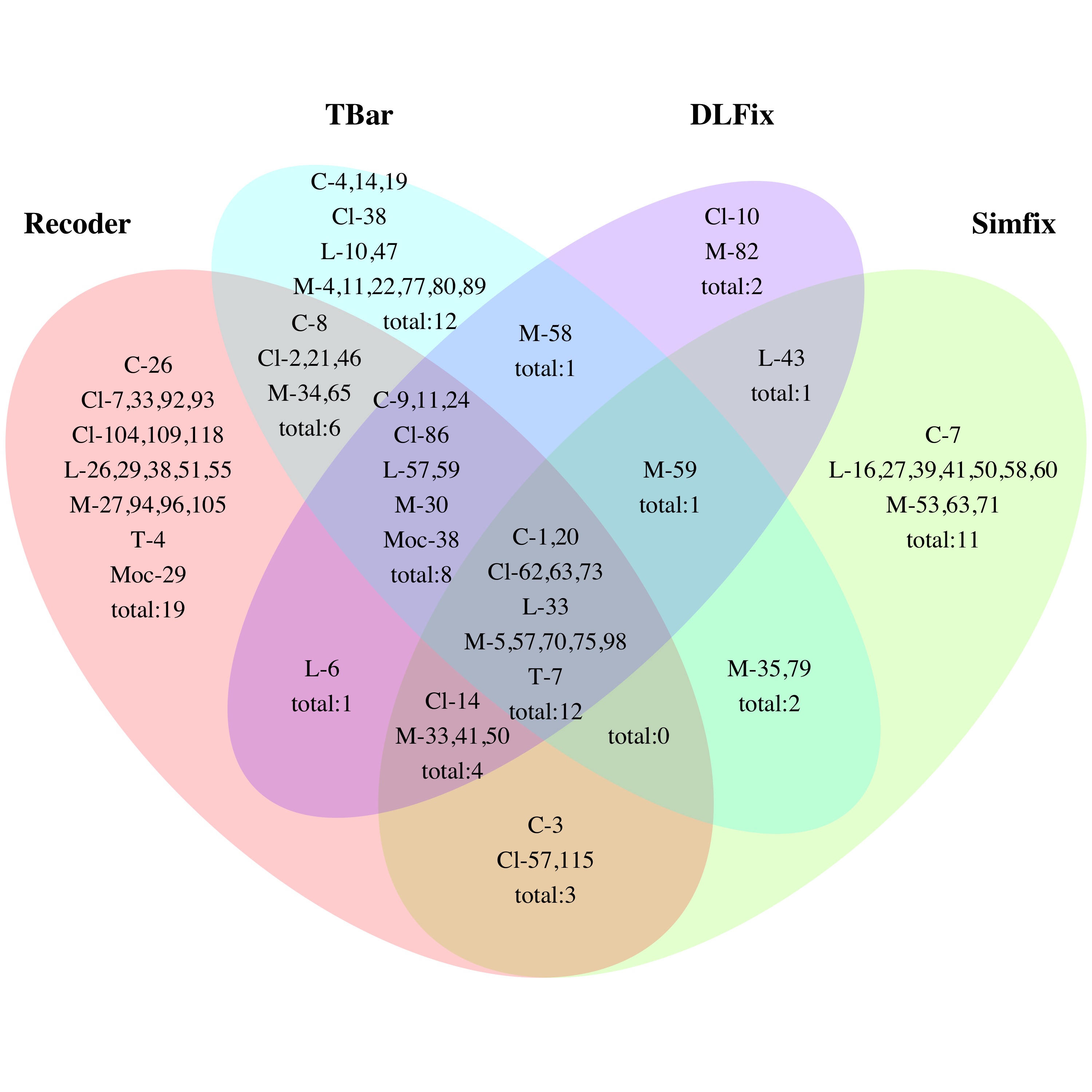}      
    {\footnotesize Project Names: C:Chart, CL:Closure, L:Lang, M:Math, Moc:Mockito, T:Time}                            
    \vspace{-3mm}
    \caption{Degree of Complementary.}
    \vspace{-6mm}
    \label{fig:overlap}
\end{figure}

We further investigate to what extent \techname complements the three best-performing existing approaches for fixing single-hunk bugs, TBar, SimFix, and DLFix. Figure~\ref{fig:overlap} reveals the overlaps of the bugs fixed by different approaches.
As shown, \techname fixes 19 unique bugs when compared with three baselines. Moreover, \techname fixes 34, 28, 27 unique bugs compared with SimFix, TBar, and DLFix, respectively. 
This result shows that \techname is complementary to these best-performing existing approaches for single-hunk bugs.
\tcbset{colframe = black}
\subsection{Contribution of Each Component (RQ2)}
To answer RQ2, we conducted an ablation test on Defects4J v1.2 to figure out the contribution of each component. Since the ablation test requires much time, we only conducted the experiment based on Ochiai Fault Localization scenario.

\begin{table}
\caption{Ablation Test for \techname on Defects4J v1.2}\vspace{-2mm}
\resizebox{\linewidth}{!}
{ 
  \begin{tabular}{l|c|c|c|c|c}
    \toprule
    Project&\tt -modify&\tt -subtreecopy&\tt -insert&\tt -placeholder&\techname\\
    \midrule
    Chart & 4 & 6 & 7 & 8 & 8\\
    Closure & 6 & 12 & 12 & 11 & 15\\
    Lang & 3 & 6 & 5 & 5 & 9\\
    Math & 7 & 8 & 9 & 9 & 15\\
    Time & 1 & 1 & 1 & 1 & 2 \\
    Mockito & 2 & 1 & 1 & 1 & 2\\
    \midrule
        Total & 23 & 34 & 35 & 35 & 51\\
  \bottomrule  
\end{tabular}
}
\label{tab:result3}
\end{table}

Table~\ref{tab:result3} shows the results of the ablation test. We respectively removed three edit operations, {\tt modify}, {\tt copy}, and {\tt insert}, as well as the generation of placeholders. As shown in the table, removing any of the components leads to a significant drop in performance. This result suggests that the two novel techniques proposed in \techname are the key to its performance.

\subsection{Generalizability of \techname (RQ3)}


\begin{table}
\vspace{-3mm}
\caption{Comparison on the 420 additional bugs}\vspace{-2mm}
\label{tab:rede2}
\resizebox{\linewidth}{!}{
  \begin{tabular}{l|c|c|c|c|c}
    \toprule
    Project& \# Used Bugs & Bug IDs & TBar & SimFix & \techname\\
    \midrule
    Cli & 39 & 1-5,7-40& 1/7 & 0/4 & \textbf{3/3}\\
    Clousre & 43 &134 - 176&0/5&\textbf{1/5}&0/7\\
    JacksonDatabind&112&1-112&0/0&0/0&0/0\\
    Codec &18 & 1-18& \textbf{2/6} & 0/2 & \textbf{2/2}\\
    Collections& 4 & 25-28& 0/1 &0/1 & 0/0 \\
    Compress & 47 & 1-47 & 1/13 &0/6 & \textbf{3/9}\\
    Csv& 16 & 1-16  & 1/5 &0/2  & \textbf{4/4} \\
    JacksonCore& 26 & 1-26 & 0/6 &0/0 & 0/4 \\
    Jsoup& 93 & 1-93 & 3/7 & 1/5 & \textbf{7/13}\\
    JxPath& 22 & 1-22 & 0/0 &0/0 & 0/4\\
    \midrule
        Total& 420 & - & 8/50 &2/25 & \textbf{19/46}\\
  \bottomrule  
\end{tabular}
}
\end{table}
\modify{Additional Experiment on RepairThemAll}{The results on Defects4J v2.0, QuixBugs and IntroClassJava are shown in Table~\ref{tab:rede2} and Table~\ref{tab:rede3}.} As shown, on Defects4J v2.0, all three approaches repair a smaller proportion of bugs, suggesting that the additional bugs on Defects4J v2.0 are probably more difficult to repair. Nevertheless, \techname still repairs most bugs compared with baselines, 19 in total, achieving 137.5\% (11 bugs) improvement over TBar and 850.0\% (17 bugs) improvement over SimFix. We believe that the considerable performance drops of TBar and SimFix are caused by their design: TBar is based on validated patterns on Defects4J v1.2, which may not generalize beyond the projects in Defects4J v1.2; SimFix relies on similar code snippets in the same project, but new projects in Defects4J v2.0 are much smaller, and thus the chance to find similar code snippets become smaller. On the other hand, \techname is trained from a large set of patches collected from different projects and is thus more likely to generalize to new projects. \modify{Additional Experiment on RepairThemAll}{On QuixBugs and IntroClassJava, \techname also repaired 775\% (31 bugs) and 30.8\% (4 bugs) more bugs on IntroClassJava and QuixBugs over the baselines respectively, further confirming the effectiveness and generalizability of \techname.}
\begin{table}
\caption{Comparison on IntroClassJava and QuixBugs}\vspace{-2mm}
\label{tab:rede3}
\resizebox{\linewidth}{!}{
  \begin{tabular}{l|c|c|c|c|c|c}
    \toprule
    Project& \# Used Bugs &jGenProg&RSRepair&Nopol&CoCoNuT & \techname\\
    \midrule
    IntroClassJava & 297 & 1/4 &4/22&3/32&- & \textbf{35/56}\\
    QuixBugs & 40 &0/3&2/4&1/4&13/20&\textbf{17/17}\\
    \midrule
        Total& 337 & 1/7 & 6/26 &4/36 &13/20& \textbf{52/73}\\
  \bottomrule  
\end{tabular}
}
\end{table}
\section{DISCUSSION}
\subsection{Adequacy of Dataset}

\begin{figure}
    \centering
    \includegraphics[width=\linewidth]{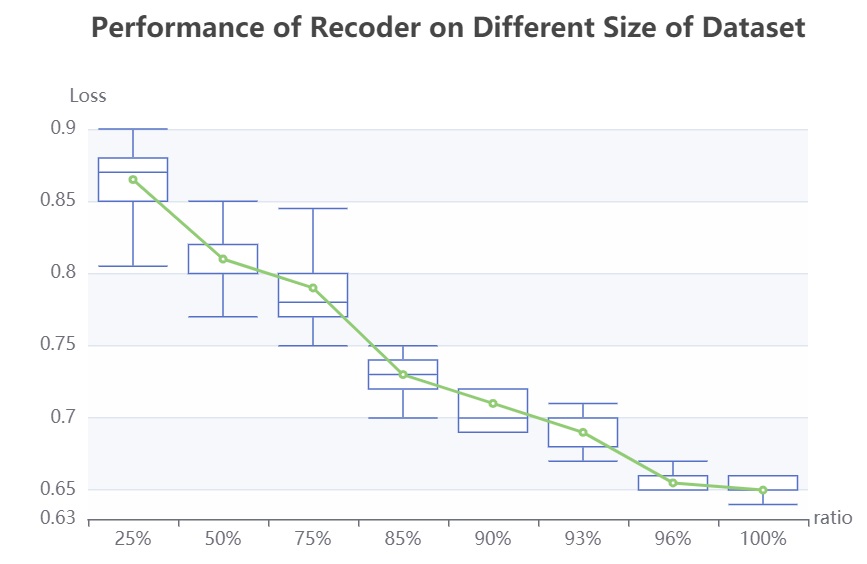}
    \vspace{-7mm}
    \caption{Adequacy of the Dataset.}
    \vspace{-2mm}
    \label{fig:tradeoff}
\end{figure}
\modify{tradeoff}{
To understand the adequacy of our training data, we trained \techname on differently sized subsets of the original training dataset and calculate the loss over the single-hunk bugs in the Defects4J v1.2 dataset. For each subset, we train 5 models with different random seeds and report the average performance of these models.
Figure~\ref{fig:tradeoff} shows the sensitivity analysis of dataset size for \techname. As shown, the loss and the diversity both decrease with the increase of training subset size. 
Therefore, the performance of \techname may further increase if more training data are provided.
}
\subsection{Limitations of \techname}
\modify{limitations of \techname}{
Recoder shares limitations common to most APR approaches: 1) Bugs to be fixed should be reproducible by failing test cases. 2) Effective fault localization is needed to identify the faulty statement. Recoder also shares the limitations common to DL-based approaches: The performance would degrade if the training set and the testing set have different distributions.}
\section{Related Work}
\paragraph{\textbf {DL-based APR Approaches}} APR has been approached using different techniques, such as heuristics or random search~\cite{genprog,6227211,10.1145/2931037.2948705,10.1145/2568225.2568254}, semantic analysis~\cite{liu2019avatar, mechtaev2015directfix,mechtaev2016angelix,mechtaev2018symbolic, gao2021beyond,chen2017contract,hua2018sketchfix,kaleeswaran2014minthint}, manually defined or automatically mined repair patterns~\cite{6606623,2018Shaping,koyuncu2020fixminer,long2017automatic,jiang2019inferring,ghanbari2019practical,bader2019getafix,rolim2017learning}, and learning from source code~\cite{2018Shaping,capgen,xiong2017precise,xiong2018learning,xin2017leveraging,long2016automatic}. For a thorough discussion of APR, existing surveys~\cite{8089448,monperrus:hal-01956501,goues2019automated,Monperrus2015} are recommended to readers.

A closely related series of work is APR based on deep learning. The mainstream approaches treat APR as a statistical machine translation that generates the fixed code with faulty code. DeepFix~\cite{10.5555/3298239.3298436} learns the syntax rules via a sequence-to-sequence model to fix syntax errors. \citet{6606623} and \citet{9000077} also adopt a sequence-to-sequence translation model to generate the patch. They use sequence-to-sequence NMT with a copy mechanism. \citet{codit} propose CODIT, which learns code edits by encoding code structures in an NMT model to generate patches. \citet{dlfix} propose a Tree-LSTM to encode the abstract syntax tree of the faulty method to generate the patches. CoCoNuT, as proposed by \citet{coconut}, adopts CNN to encode the faulty method and generates the patch token by token. Compared to them, our paper is the first work that aims to improve the decoder and employs a syntax-guided manner to generate edits, with the provider/decider architecture and placeholder generation. 

\paragraph{\textbf {Multi-Hunk APR}} Most of the above approaches generate single-hunk patches---patches that change one place of the program. Recently, \citet{saha2019harnessing} propose Hercules to repair multi-hunk bugs by discovering similar code snippets and applying similar changes. Since lifting single-hunk repair to multiple-hunk repair is a generic idea and can also be applied to \techname, we did not directly compare \techname with the multi-hunk repair tools in our evaluation. Nevertheless, though \techname only repairs single-hunk bugs, we notice that it still outperforms Hercules by repairing 5 more bugs on the Defects4J v1.0 (including all projects from v1.2 except for Mockito), the dataset Hercules has been evaluated on.

\paragraph{\textbf{DL-based Code Generation}} Code generation aims to generate code from a natural language specification and has been intensively studied during recent years. With the development of deep learning, \citet{ling2016latent} propose a neural machine translation model to generate the program token by token. Being aware that code has the constraints of grammar rules and is different from natural language, \citet{DBLP:conf/acl/YinN17} and \citet{DBLP:conf/acl/RabinovichSK17} propose to generate the AST of the program via expanding from a start node. To integrate the semantic of identifiers, OCoR~\cite{ocor} proposes to encode the identifiers at character level. To alleviate the long dependency problem, a CNN decoder~\cite{DBLP:conf/aaai/SunZMXLZ19} and TreeGen (a tree-based Transformer)~\cite{treegen} are proposed to generate the program. In this paper, we significantly extend TreeGen to generate the edit sequence for program repair.

\paragraph{\textbf{DL-based Code Edit Generation}}
Several existing DL-based approaches also use the idea of generating edits on programs~\cite{yasunaga2020graphbased,DBLP:journals/corr/abs-1911-01205,structuralChange,Dinella2020HOPPITY, cao2020}. \citet{DBLP:journals/corr/abs-1911-01205} view a program as a sequence of tokens and generate a sequence of token-editing commands. \citet{structuralChange} view a program as a tree and generate node-editing or subtree-editing commands. \citet{Dinella2020HOPPITY} view a program as a graph and generate node-editing commands. Compared with our approach, there are three major differences. First, the existing approaches are not syntax-guided and treat an edit script as a sequence of tokens. As a result, they may generate syntactically incorrect edit scripts and do not ensure syntactic correctness of the edited program. On the other hand, our approach introduces the provider/decider architecture and successfully realizes the syntax-guided generation for edits. Second, none of the existing approaches support placeholder generation, and thus are ineffective in generating edits with project-specific identifiers. Third, the editing commands they use are atomic and are inefficient in representing large changes. For example, to insert a variable declaration, in our approach there is one {\tt insert} operation: {\tt insert(int var = 0;)}. However, the existing approaches have to represent this change as a sequence of 5 insertions, where each insertion inserts one token. 


\section{THREATS TO VALIDITY}
\paragraph{\textbf{Threats to external validity}} mainly lie in the evaluation dataset we used. First, though our approach applies to different programming languages, so far, we have only implemented and evaluated it on Java, so future work is needed to understand its performance on other programming languages. Second, though we have evaluated on Defects4J v2.0, QuixBugs, and IntroClassJava, it is yet unknown how our approach generalizes to different datasets~\cite{JiangLNZH21}. This is a future work to be explored. 
\paragraph{\textbf{Threats to internal validity}} mainly lie in our manual assessment of patch correctness. To reduce this threat, two authors have independently checked the correctness of the patches, and a patch is considered correct only if both authors consider it correct. The generated patches also have been released for public assessment.

\section{CONCLUSION}
\modify{}{In this paper, we propose \techname, a syntax-guided edit decoder with placeholder generation for automated program repair. 
\techname uses a novel provider/decider architecture to ensure accurate generation and syntactic correctness of the edited program and generates placeholders for project-specific identifiers.
In the experiment, \techname achieved 21.4\% improvement (9 bugs) over the existing state-of-the-art APR approach for single-hunk bugs on Defects4J v1.2. Importantly, \techname is the first DL-based APR approach that has outperformed traditional APR techniques on this benchmark. Further evaluation on three other benchmarks shows that \techname has better generalizability than some state-of-the-art APR approaches.}

\section*{ACKNOWLEDGMENTS}

\modify{acknowledgment}{This work is sponsored by the National Key Research and Development Program of China under Grant No. 2017YFB1001803, National Natural Science Foundation of China under Grant Nos. 61922003, and a grant from ZTE-PKU Joint Laboratory for Foundation Software.}

\balance
\bibliographystyle{ACM-Reference-Format}
\bibliography{sample-base}

\appendix

\end{document}